\documentclass[prl,superscriptaddress,showpacs,twocolumn]{revtex4}

\usepackage{latexsym,epsfig,epstopdf}
\usepackage{amssymb,amsmath,amsthm}

\usepackage{times}
\usepackage{color}

\newcommand{\ket}[1]{\left | #1 \right\rangle}
\newcommand{\bra}[1]{\left \langle #1 \right |}

\newcommand{\Tr}{\mathrm{Tr}}

\newcommand{\proj}[1]{\ket{#1}\bra{#1}}

\newcommand{\cE}{\mathcal{E}}
\newcommand{\cD}{\mathcal{D}}
\newcommand{\cM}{\mathcal{M}}

\newtheorem{lemma}{Lemma}
\newtheorem{theorem}{Theorem}
\newtheorem{example}{Example}

\begin{document}

\title{Quantum discord bounds the amount of distributed entanglement}

\author{T. K. Chuan}
\affiliation{Centre for Quantum Technologies, National University of Singapore, 3 Science Drive 2, 117543 Singapore}

\author{J. Maillard}
\affiliation{Blackett Laboratory, Imperial College London, Prince Consort Road, London SW7 2BZ, UK}

\author{K. Modi}
\affiliation{Department of Physics, University of Oxford, Clarendon Laboratory, Oxford, OX1 3PU, UK}
\affiliation{Centre for Quantum Technologies, National University of Singapore, 3 Science Drive 2, 117543 Singapore}

\author{T. Paterek}
\email{tomasz.paterek@ntu.edu.sg}
\affiliation{Centre for Quantum Technologies, National University of Singapore, 3 Science Drive 2, 117543 Singapore}

\author{M. Paternostro}
\affiliation{Centre for Theoretical Atomic, Molecular, and Optical Physics, School of Mathematics and Physics, Queen's University, Belfast BT7 1NN, UK}

\author{M. Piani}
\email{mpiani@iqc.ca}
\affiliation{Institute for Quantum Computing \& Department of Physics and Astronomy, University of Waterloo, 200 University Avenue West, N2L 3G1 Waterloo, Ontario, Canada}

\begin{abstract}
The ability to distribute quantum entanglement is a prerequisite for  many fundamental tests of quantum theory and numerous quantum information protocols. Two distant parties can increase the amount of entanglement between them by means of quantum communication encoded in a carrier that is sent from one party to the other. Intriguingly, entanglement can be increased even when the exchanged carrier is not entangled with the parties. However, in light of the defining property of entanglement stating that it cannot increase under classical communication, the carrier must be quantum.
Here we show that, in general, the increase of relative entropy of entanglement between two remote parties is bounded by the amount of non-classical correlations of the carrier with the parties as quantified by the relative entropy of discord.
We study implications of this bound, provide new examples of entanglement distribution via unentangled states and put further limits on this phenomenon.
\end{abstract}


\maketitle

\emph{Introduction.---}Entanglement is a trademark of quantum physics~\cite{ENT_REV} and a powerful resource enabling faster-than-classical computation~\cite{ONE-WAY-REVIEW}, efficient quantum communication~\cite{CCP_REV} and secure cryptography~\cite{EKERT1991}. 
For these reasons, the design of efficient methods to distribute entanglement
is one of the key goals of mainstream quantum information science. Of particular relevance for tasks of long-haul quantum communication is the distribution of entanglement among the remote non-interacting nodes of a quantum network~\cite{KIMBLE}. In this case, two general architectures able to accomplish this task have been identified: the first relies on the availability of a resource whose entanglement is {\it transferred} to chosen nodes of the network~\cite{TRANSF,TRANSF2,ENT_SWAPPING}. The second is {a quantum communication scenario} based on the exchange of a {\it carrier} quantum system between two of such distant nodes~\cite{SHUTTLE}, which might be referred to as the {\it sender} and {\it receiver} laboratory, respectively.

Remarkably, Cubitt \emph{et al.}~\cite{CVDC} reported a scheme where the carrier exchanged by sender and receiver remains unentangled from them at all times. This result, which was later extended to the continuous-variable scenario in~\cite{MK,MK2}, intriguingly implies that the amount of distributed entanglement does not appear to be bounded by the entanglement initially shared by the carrier and  the sender, given that in these cases they are unentangled at all times. These observations pave the way to some interesting considerations. First, quite clearly, the carrier must display some quantum features, otherwise the protocol would simply consist of the exchange of classical communication aided by local node-carrier operations, which cannot increase entanglement~\cite{LOCC}. Second, in Refs.~\cite{DATTA_THESIS,LL2008} a link has been suggested between the distribution of entanglement by separable states and the presence of more general forms of quantum correlations, as captured for example by quantum discord~\cite{HV,OZ, DISCORD_REV}, between nodes of the network and the carrier.

In light of such considerations, here we address the following fundamental questions: {\it How much can the entanglement between sender and receiver laboratories  increase under the exchange of a carrier? Is there a quantitative relation between such increase and the non-classical correlations between the carrier and the parties?}

{The key finding of our work is a general bound on how much entanglement can increase under local operations and quantum communication: the entanglement gain between distant laboratories is bounded by the amount of quantum discord between them and the carrier.} In turn, this result provides an operational interpretation of quantum discord as the truly necessary prerequisite for the success of entanglement distribution {as opposed to entanglement itself}. We show that the relation thus formulated generalizes the subadditivity of entropy and can be quite naturally linked to the possibility that quantum conditional entropy attains negative values~\cite{MERGING, THERMO_S}. Finally, we study in detail the resources required for entanglement creation and increase via the use of a separable carrier, and illustrate our findings with some new concrete examples of such phenomenon.

\emph{Definitions.---}In order to treat entanglement and discord on the same footing, throughout this paper we consider
the former as measured by the relative entropy of entanglement~\cite{REL_ENT_E,REL_ENT_E2} and the latter as quantified by the one-way quantum deficit~\cite{DEFICIT}, also known as relative entropy of discord~\cite{MODI2010}.
The quantum relative entropy between two states $\rho$ and $\sigma$ is defined as $S(\rho\|\sigma):=-S(\rho)-\Tr(\rho\log\sigma)$.
It is monotone under any completely positive trace-preserving map $\cM$, that is $ S(\rho\|\sigma)\geq S(\cM(\rho)\|\cM(\sigma))$. The relative entropy of entanglement in the bipartition $X$-versus-$Y$ is defined as the minimum relative entropy $\mathcal{E}_{X:Y}(\rho):=\min_{\rho_{X:Y}} S(\rho\|\rho_{X:Y})$ between the joint state $\rho$ of $X$ and $Y$ and the set of separable states $\rho_{X:Y}=\sum_i p_i \rho_X^i\otimes\rho_Y^i$~\cite{REL_ENT_E,REL_ENT_E2}. 
Similarly, the relative entropy of discord is defined as the minimum relative entropy $\mathcal{D}_{X|Y}(\rho):=\min_{\chi_{X|Y}} S(\rho||\chi_{X|Y})$ between $\rho$ and the set of quantum-classical states $\chi_{X|Y} = \sum_j p_j~ \chi^j_{X} \otimes \proj{j}_Y$, with $\{\ket{j}\}$ an orthonormal basis for $Y$.
It can be shown that $\mathcal{D}_{X|Y}(\rho)$ corresponds to the minimal entropic increase resulting from the performance of a complete projective measurement $\Pi_Y$ over $Y$: $\mathcal{D}_{X|Y}(\rho){=}\min_{\Pi_Y}S(\Pi_Y(\rho)){-}S(\rho)$ 
where $\Pi_Y(\rho)$
describes the state after the measurement $\Pi_Y$~\cite{MODI2010}.
Finally, mutual information between $X$ and $Y$ is defined as $\mathcal{I}_{X:Y}(\rho):=S(\rho_{XY}\|\rho_X\otimes \rho_Y)$, with $\rho_X$ and $\rho_Y$ the reduced states of $X$ and $Y$.
Mutual information quantifies the total amount of correlations present between $X$ and $Y$~\cite{GPW2005}. It holds $\mathcal{I}_{X:Y}(\rho) \geq\cD_{X|Y}(\rho)\geq \cE_{X:Y}(\rho)$.

\begin{figure}
\label{FIG_DIST}
\includegraphics[width=0.8\columnwidth]{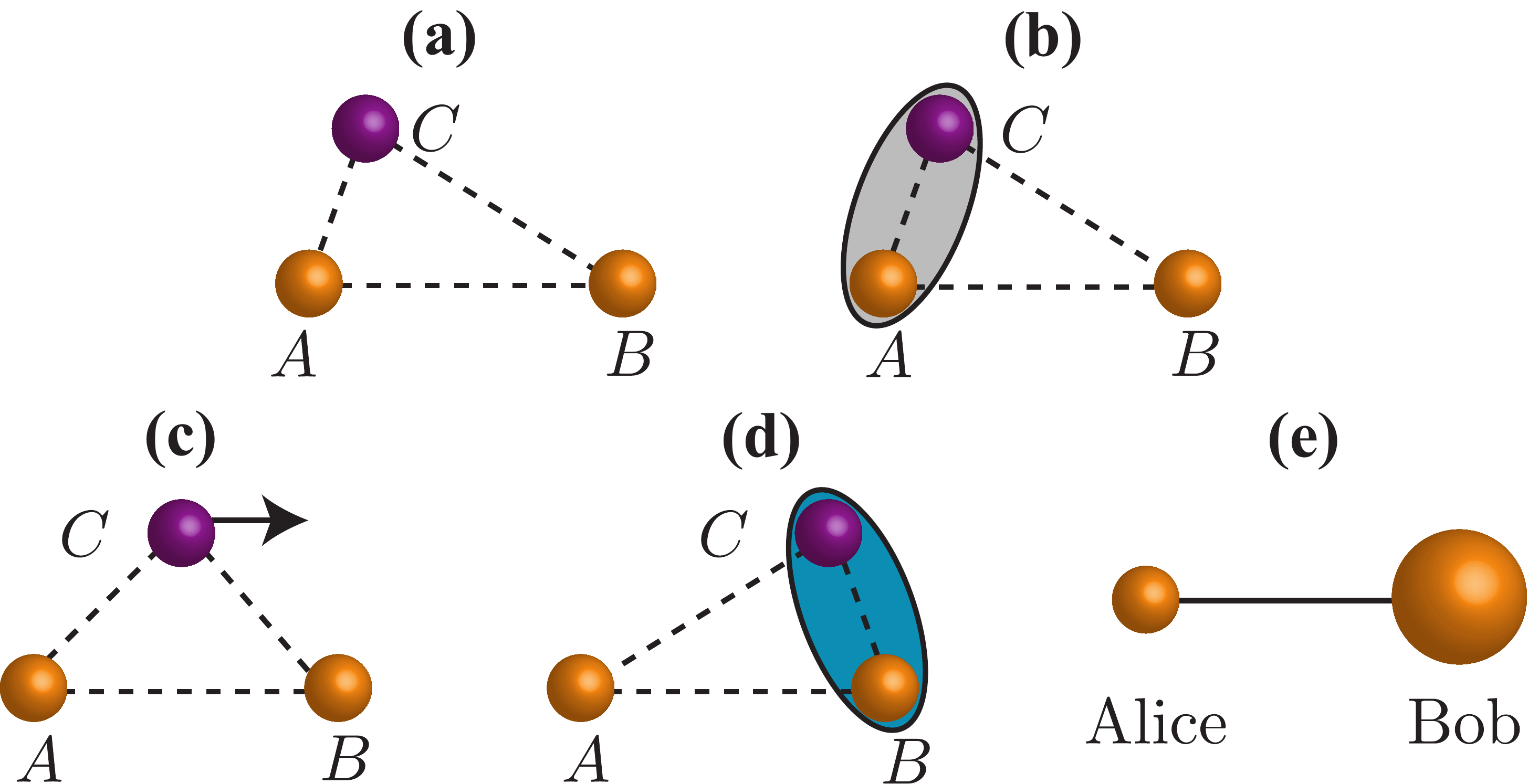}
\caption{Entanglement distribution.
{\bf (a)} The distribution protocol begins with systems $A$ and $C$ in Alice's lab and system $B$ in Bob's.
{\bf (b)} In the next step, Alice applies an {\it encoding} operation to systems $A$ and $C$.
{\bf (c)} System $C$ is then sent to Bob's site.
{\bf (d)} The carrier $C$ interacts with $B$ via a {\it decoding} operation meant to localize on $B$ the entanglement between $A$ and $BC$.
{\bf (e)} Systems $A$ and $B$ are more entangled than in panel (a).}
\end{figure}

\emph{Entanglement distribution.---}Consider two remote agents, Alice and Bob, having access to local quantum systems $A$ and $B$, respectively.
Their aim is to increase the entanglement that they share by sending an auxiliary quantum system---the carrier $C$---with which they interact locally (see Fig.~1).
The key step of any communication scheme is the transfer of a carrier system from one laboratory to the other.
The difference in entanglement across the two bipartitions $A:CB$ and $AC:B$,
corresponding to the situation after and before the transfer of the carrier, can be bound thanks to the following (see the Appendix for a proof)
\begin{theorem}
\label{TH1}
For any tripartite state $\rho=\rho_{ABC}$ it holds
\begin{equation}
\label{FOR_DELTA}
|\cE_{A:CB}(\rho)-\cE_{AC:B}(\rho)|\leq \cD_{AB|C}(\rho).
\end{equation}
\end{theorem}
We apply this relation to the scenario of Fig.~1. Let us call $\alpha$ the initial state of $A,B$ and $C$, and $\beta=\mathcal{M}_{AC}(\alpha)$ the state obtained from it by means of a local \emph{encoding} operation $\mathcal{M}_{AC}$ that does not increase entanglement in the $AC:B$ cut, i.e.  $\cE_{AC:B}(\beta)\leq \cE_{AC:B}(\alpha)$. System $C$ is then sent to Bob's site, where it interacts with $B$ via a {\it decoding} operation meant to localize on $B$ alone the entanglement between the laboratories~\cite{ENT_LOC}.
As a side note, we mention that one could also consider local encoding operations that add ancillary systems. However, this is taken into account by including all ancillas in $A$ or $C$ from the very beginning. 
Combining the above description with Eq.~\eqref{FOR_DELTA} for $\beta$ we arrive at
\begin{equation}
\label{FOR_DELTA_E}
\cE_{A:CB}(\beta) 
\leq\cE_{AC:B}(\alpha)+\cD_{AB|C}(\beta).
\end{equation}
This shows that the entanglement gain between distant laboratories is bounded by the amount of quantum discord as measured on the communicated system -- the communicated quantum correlations.
In what follows we discuss the meaning and the implications of the bounds given in Eqs.~\eqref{FOR_DELTA} and \eqref{FOR_DELTA_E}.

\emph{Impossibility of entanglement distribution by local operations and classical communication.---}Let us first address the case of ${\cal D}_{AB|C}(\beta){=}0$.
This corresponds to classical communication from Alice to Bob as it implies that $\beta$ has the quantum-classical structure $\beta = \sum_{i} p_i \rho^i_{AB} \otimes \proj{i}_{C}$ . The index $i$ embodies classical information that Alice may copy locally before sending $C$ to Bob.
After $C$ is transferred from Alice to Bob, both have access to this information.
Bob can then perform a local transformation that depends on the index $i$ originally held only by Alice.  The process just described is one communication step of a general protocol based on the use  of local operations and classical communication (LOCC).
The protocol may include several rounds of classical communication with $C$ that is sent back and forth between Alice and Bob; local classical registers can be kept or erased at any stage of the protocol. {In this case, Eq.~\eqref{FOR_DELTA_E} reduces to the statement that entanglement does not increase at any step of a protocol based on LOCC~\cite{LOCC}. If ${\cal D}_{AB|C}(\beta)$ does not vanish, the transfer of $C$ cannot be interpreted as classical communication revealing the role of discord in general quantum communication.
Hence, Eq.~\eqref{FOR_DELTA_E} constitutes a non-trivial relaxation of the condition of monotonicity of entanglement under LOCC, bounding the increase of entanglement under local operations and \emph{quantum communication}.}

\emph{Subadditivity of entropy.---}Let us now take a tripartite pure state $\rho=\proj{\phi}_{ABC}$. 
Since for a generic pure state $\ket{\psi}_{XY}$ both the relative entropy of entanglement and the relative entropy of discord coincide with the entropy of the reduced states of $X$ or $Y$,
Eq. \eqref{FOR_DELTA} becomes
\begin{equation}
\label{EQ:ALINEQ}
|S(\rho_{A})-S(\rho_B)|\leq S(\rho_{AB}),
\end{equation}
which is the Araki-Lieb inequality for the von Neumann entropy~\cite{ARAKI-LIEB} and is equivalent to the subadditivity of entropy for subsystems $AC$ and $BC$. {Accordingly, Eq.~\eqref{FOR_DELTA} can be interpreted as a possible generalization of the subadditivity of entropy, based on the concepts of entanglement and quantumness of correlations and valid for tripartite mixed states.}

\emph{Simple meaning of quantum conditional entropy.---}Consider the bipartite system composed of $A$ and $C$, both held at Alice's location, and prepared in a state $\rho_{AC}$ with conditional entropy $S_{C|A}(\rho) := S(\rho_{AC})-S(\rho_A)$.
Let us introduce a third system $B$ being a purification of $\rho_{AC}$ and let us place it in a distant laboratory.
The left-hand side of Eq.~\eqref{FOR_DELTA}, written for a pure tripartite system, reads
\begin{equation}
\cE_{A:CB}(\rho)-\cE_{AC:B}(\rho) = S_{C|B}(\rho) = - S_{C|A}(\rho).
\end{equation}
Therefore, the negative conditional entropy $- S_{C|A}$ of $\rho_{AC}$ gives the increase of entanglement between distant laboratories caused by the transfer of $C$.

\emph{Entanglement distribution via separable system.---}The bound derived in Eq.~\eqref{FOR_DELTA} is tight in some cases; in particular, we have verified that it is tight for the three-qubit state of the seminal example of entanglement creation with an unentangled carrier introduced in Ref.~\cite{CVDC}.
Motivated by this, and in order to emphasize the significance of the appearance of discord rather than entanglement on the right-hand side of Eq.~\eqref{FOR_DELTA},
here we focus on the general conditions for the success of entanglement creation by means of a separable carrier. In the present framework, this corresponds to requiring
\begin{subequations}
\label{NECC_FOR_DIST}
\begin{align}
\mathcal{E}_{B:AC}(\alpha) &= 0 \quad(\Rightarrow \mathcal{E}_{B:AC}(\beta)= 0),\label{EQ:NECB}\\
\mathcal{E}_{C:AB}(\beta) &= 0,\label{EQ:NECC}\\
\mathcal{E}_{A:BC}(\beta) &> 0\label{EQ:NECA}.
\end{align}
\end{subequations}
Eq.~\eqref{EQ:NECB} says that no entanglement between the distant sites is present initially. The implication is due to the local nature of the encoding operation $\cM_{AC}$. Eq.~\eqref{EQ:NECC} encompasses our prescription that the carrier must be separable from $A$ and $B$. Finally, Eq.~\eqref{EQ:NECA} ensures that non-vanishing entanglement is established by exchanging the carrier. We remark that non-vanishing  $A:BC$ entanglement does not necessarily imply the possibility of creating $A:B$ entanglement via the local decoding operation on $BC$ mentioned above. Indeed, if this was always possible, bound entanglement~\cite{BOUND_ENT} would not exist, as one could always map entanglement into two-qubit entanglement, which is known to be distillable~\cite{2QUBIT_DIST}. However, in many relevant cases, including all our examples, entanglement can be localized as shown by the Theorem 2 in the Appendix.

In order to satisfy the conditions \eqref{NECC_FOR_DIST}, besides the discord present in $\beta$, there must be discord on the receiver side already in the initial state $\alpha$. This is seen by applying Eq.~\eqref{FOR_DELTA} again, but with the roles of $B$ and $C$ interchanged, and using the fact that discord does not increase under operations on the unmeasured systems~\cite{DISCORD_ENT}, arriving to
\begin{equation}
\label{EQ:BQUANTUM}
\cE_{A:CB}(\beta)\leq\cE_{AB:C}(\beta) 
+\cD_{AC|B}(\alpha).
\end{equation}
If Eq.~\eqref{EQ:NECC} holds, we obtain the relation $\mathcal{E}_{A:BC}(\beta)\leq\mathcal{D}_{AC|B}(\alpha)$.
Note that if $C$ is initially not correlated with $AB$, tha latter further simplifies to $\mathcal{E}_{A:BC}(\beta)\leq\mathcal{D}_{A|B}(\alpha)$.
Another interesting limiting case of  Eq.~\eqref{EQ:BQUANTUM} is when $\cD_{AC|B}(\alpha)=0$. Then $B$ is classical initially and therefore also in the state $\beta$ after the encoding:
$\beta=\sum_ip_i\beta^i_{AC}\otimes\proj{i}_B$. In this case entanglement between Alice and Bob can only be created if the carrier is entangled with the sites and, in particular, only if at least one $\beta^i_{AC}$ is entangled. Indeed, such $\beta$ simply describes a situation in which Bob, upon reading the index $i$ encoded in $B$, knows which of many states $\beta^i_{AC}$ he will end up sharing with Alice.

On the other hand, entanglement creation with a separable carrier is  possible starting from a state with $\cD_{BC|A}(\alpha)=0$. For instance, it is enough to consider the three-qubit example given in Ref.~\cite{CVDC}, but starting with $A$ and $C$ interchanged and using a step in the encoding operation $\cM_{AC}$ to undo the change before proceeding with the original protocol.
However, under further restrictions, the classicality of $A$ may  
prevent entanglement creation with a separable carrier, as shown for instance in Theorem 3 in the Appendix.


Furthermore, we note that when the encoding operation is restricted to be unitary, the presence of discord (on either party) is not a sufficient precondition to make entanglement creation with a separable carrier possible. This follows by combining the fact that any bipartite state that is sufficiently mixed is separable~\cite{GB} and the existence of discordant states infinitesimally close to any non-discordant one~\cite{FACCA}. As unitary operations do not change mixedness, discord of sufficiently mixed states cannot be converted into entanglement.

Finally, for a fixed dimension of the carrier, it is more efficient to use an entangled carrier rather than a separable one. On one hand, by sending a $d$-dimensional system  that is maximally entangled with a similar one that remains with the sender, we can increase the shared entanglement by $\log_2 d$. On the other hand, Theorem 4 of the Appendix shows that using separable states the entanglement increase is strictly smaller than $\log_2 d$.

\emph{Examples.---}In order to make our result more concrete, in Appendix we provide new examples of both the creation and the increase of entanglement between distant parties by the exchange of an unentangled carrier.
The examples are based on the fact that the state of a bipartite system of total dimensions $d_{\mathrm{tot}}$ having the form
$\rho_{p} = p \ket{\psi} \bra{\psi} + (1-p){\openone}/{d_{\mathrm{tot}}}$
 is separable if and only if $p \le p_{\mathrm{cr}}=(1+ a_1 a_2 d_{\mathrm{tot}})^{-1}$, where $a_{1}$ and $a_2$ are the two largest Schmidt coefficients of the bipartite state $\ket{\psi}$, and $\openone/d_{\mathrm{tot}}$ is the maximally mixed state of the total system~\cite{VT1999}. Consider now a tripartite pure state $\ket{\psi}=\ket{\psi}_{ABC}$. This state admits three Schmidt decompositions corresponding to the three bipartitions $A:BC$, $B:AC$, and $C:AB$. One can choose $\ket{\psi}$ such that $p_{\mathrm{cr}}$ is the lowest across the $A:BC$ bipartition, so that there is a finite range for $p$ such that $\rho_{p}$ is $A:BC$-entangled but separable in the remaining two splittings. Such a $\rho_p$ is meant to play the role of $\beta$ in our scenario. 
We remark that the three-qubit example of Ref.~\cite{CVDC} uses a carrier system $C$ that is initially classically correlated with $A$ and $B$. However, a scenario where $C$ initially shares no correlation with the remote nodes is more relevant from a practical  point of view, as one can imagine that the carrier is an independent system to be used to distribute entanglement. Even with such a restriction, entanglement can be established via a separable system, as proven explicitly by our examples in the Appendix.

\emph{Conclusions.---}
It is the very act of physical transmission of a carrier system that changes the amount of correlations between the remote laboratories.
To illustrate this consider total correlations, as captured by mutual information.
One expects from the principle of no-signaling or information causality~\cite{IC} that the increase of mutual information is bounded by the amount of communicated correlations.
Indeed, applying the chain rule for mutual information and its monotonicity under local operations~\cite{NC} one finds
\begin{equation}
\label{EQ:MINFO}
\mathcal{I}_{A:CB} - \mathcal{I}_{AC:B} \le \mathcal{I}_{A:C} \le \mathcal{I}_{AB:C}.
\end{equation}
Both in classical and quantum information theory, the increase of total correlations between the labs is bounded by the correlations between the systems that are kept stored in the labs and the carrier -- the communicated total correlations.

However, whereas there is only one kind of correlations between classical random variables, quantum systems can share different kinds of correlations~\cite{DISCORD_REV}. 
In this work we proved a relation analogous to Eq.~\eqref{EQ:MINFO} for the increase of quantum entanglement between remote elements of a quantum network.
We showed that such increase is bounded from above by the amount of non-classical correlations between the exchanged carrier and the distant nodes
as measured by quantum discord, a quantifier for a more general type of non-classical correlations than entanglement.
It follows that, in contrast with what one would expect extrapolating from Eq.~\eqref{EQ:MINFO}, our bound for the entanglement increase is in general larger than the entanglement between the carrier and the nodes; in particular, it can be non-zero even when the latter vanishes. 
Indeed, this has to be the case, as implied by the seminal example of entanglement distribution using a separable carrier of Ref.~\cite{CVDC}.

Besides providing a natural operational interpretation of quantum discord as the truly necessary prerequisite for the success of entanglement distribution, our work identifies the conditions for the occurrence of entanglement distribution with a separable carrier. The scenario tackled by our study is general enough to fit well with a few experimental settings, including cavity/circuit-QED and trapped-ion technology and we thus hope that our results will find a prompt experimental demonstration.

\emph{Acknowledgments.---}M. Paternostro thanks the Centre for Quantum Technologies, National University of Singapore and, together with J. Maillard, the Institute for Quantum Computing, University of Waterloo for the kind hospitality during the early stages of this work. We acknowledge financial support from the National Research Foundation and Ministry of Education in Singapore (T. K. Chuan, K. Modi, and T. Paterek), the John Templeton Foundation (K. Modi), the UK EPSRC (M. Paternostro), NSERC, CIFAR, and the Ontario Centres of Excellence (M. Piani). M. Paternostro acknowledges valuable discussions with F. Ciccarello during the early stages of this work.

\emph{Note added.---}During the completion of this work we became aware of the closely related independent work by A. Streltsov \emph{et al.} ~\cite{ALEXANDUS}. 

\section{Appendix}

In this Appendix we provide statements and proofs of the theorems mentioned in the main text and present new examples of both the creation and the increase of entanglement between distant parties by the exchange of an unentangled carrier.

\subsection{Theorems}

We prove here Theorem \ref{TH1} of the main text.
It is a consequence of the following Lemma.
\begin{lemma}
\label{LEM:NEWLEMMA}
Given $\rho=\rho_{ABC}$, consider $\Pi^*_C$, the optimal projective measurement on $C$ for the sake of $\cD_{AB|C}(\rho)$. Let $p_i$ be the probability of outcome $i$ for such a measurement, and $\rho^i_{AB}$ be the corresponding conditional states of $AB$; i.e., $\Pi^*_C(\rho_{ABC})=\sum_i p_i \rho^i_{AB}\otimes \proj{i}_C$. Then
\begin{equation}
\label{EQ:NEWLEMMA}
\begin{split}
\cE_{A:CB}(\rho) &\leq \cD_{AB|C}(\rho) +   \sum_ip_i\cE_{A:B}(\rho^i_{AB})\\
  &= \cD_{AB|C}(\rho) + \cE_{A:CB}(\Pi^*_C(\rho))\\
  &= \cD_{AB|C}(\rho) + \cE_{AC:B}(\Pi^*_C(\rho))
\end{split}
\end{equation}
\end{lemma}
\begin{proof}
Let $\rho^{i*}_{A:B}$ be the optimal separable state for the sake of $\cE_{A:B}(\rho^i_{AB})$. The state $\sum_i p_i \rho^{i*}_{A:B}\otimes \proj{i}_C$ is fully separable and \emph{a fortiori} $A:CB$-separable; moreover it is invariant under the action of $\Pi^*_C$. 
Then the inequality \eqref{EQ:NEWLEMMA} is obtained as follows:
\begin{subequations}
\begin{align}
&\quad\cE_{A:CB}(\rho) \leq S(\rho\|\sum_i p_i \rho^{i*}_{A:B}\otimes \proj{i}_C ) \label{eq:thm1.2}\\
&=-S(\rho)-\Tr[ \rho \log (\sum_i p_i \rho^{i*}_{A:B}\otimes \proj{i}_C ) ]\notag\\
&=-S(\rho)-\Tr[ \Pi^*_C(\rho) \log (\sum_i p_i \rho^{i*}_{A:B}\otimes \proj{i}_C )]\label{eq:thm1.3}\\
&=\Big[S(\Pi^*_C(\rho))-S(\rho)\Big]\notag+\Big[-S(\Pi^*_C(\rho))\\
&\quad-\Tr\Big(\Pi^*_C(\rho) \log ( \sum_i p_i \rho^{i*}_{A:B}\otimes \proj{i}_C) \Big)\Big]\notag\\
&=\cD_{AB|C}(\rho)\notag\\
&\quad+S(\sum_i p_i \rho^i_{AB}\otimes \proj{i}_C\|\sum_i p_i \rho^{i*}_{A:B}\otimes \proj{i}_C)\label{eq:thm1.4}\\
&= \cD_{AB|C}(\rho)+\sum_ip_iS(\rho^i_{AB}\| \rho^{i*}_{A:B})\label{eq:thm1.5}\\
&=\cD_{AB|C}(\rho) + \sum_ip_i\cE_{A:B}(\rho^i_{AB}),\label{eq:thm1.6}
\end{align}
\end{subequations}
where the steps are justified as follows: for Eq.~\eqref{eq:thm1.2},  the fully separable state $\sum_i p_i \rho^{i*}_{A:B}\otimes \proj{i}_C$ cannot be better than optimal for the sake of $\cE_{A:CB}(\rho)$; for Eq.~\eqref{eq:thm1.3}, $\Tr(\sigma\log\Pi(\tau))=\Tr(\Pi(\sigma)\log\Pi(\tau))$ for all (complete or non-complete) projective measurements $\Pi$, and for all $\sigma$ and all $\tau$~\cite{NC}; for Eq.~\eqref{eq:thm1.4}, by the optimality of $\Pi^*_C$ for the sake of $\cD_{AB|C}(\rho)$; for Eq.~\eqref{eq:thm1.5}, by the chain rule for relative entropy~\cite{restrictedmeasurements}; for Eq.~\eqref{eq:thm1.6}, by the optimality of each $\rho^{i*}_{A:B}$ for the sake of $\cE_{A:B}(\rho^i_{AB})$.
 Finally, the two last lines of Eq.~\eqref{EQ:NEWLEMMA} are due to the fact that relative entropy of entanglement satisfies the ``flags'' condition of Ref.~\cite{MHflags}, i.e. $\cE_{FX:Y}\left(\sum_i p_i \proj{i}_F\otimes\rho^i_{XY}\right)=\sum_i p_i \cE_{X:Y}(\rho^i_{XY})= \cE_{X:YF}\left(\sum_i p_i\rho^i_{XY}\otimes \proj{i}_F\right)$.
\end{proof}
The statement of the above Lemma regards entanglement redistribution. Nonetheless it is related to --- and can be seen as a generalization of --- the results of Ref.~\cite{LOCKINGENT}, where it was proven that the relative entropy of entanglement is not lockable by dephasing any single qubit held by one of the parties.
In our context, it is further worth recalling that  the variation of a generic relative entropy-based measure of correlations --- not necessarily entanglement --- under the complete dephasing of one of the two parties was considered in Ref.~\cite{DEFICIT}.
We notice that the total dephasing of one of the two parties would simply destroy all entanglement. The bound given in Eq.~\eqref{EQ:NEWLEMMA} is based on the consideration of a hypothetical optimal complete von Neumann measurement performed only on the subsystem that is to be transferred from one party to the other.

\noindent\emph{Proof of Theorem \ref{TH1}}. 
Applications of Lemma~\ref{LEM:NEWLEMMA} and the monotonicity of the relative entropy of entanglement under LOCC gives
\begin{equation}
\begin{split}
\cE_{A:CB}(\rho) &\leq\cD_{AB|C}(\rho) + \cE_{AC:B}(\Pi^*_C(\rho))\\
&\leq \cD_{AB|C}(\rho) + \cE_{AC:B}(\rho).
\end{split}
\end{equation}
By inverting the roles of $A$ and $B$, we obtain Eq.~\eqref{FOR_DELTA}.\hfill$\Box$

We remark that Lemma~\ref{LEM:NEWLEMMA}, although less amenable to a clear operational interpretation, is in general strictly stronger than Theorem \ref{TH1}. Consider for example the case of a pure tripartite state symmetric under the exchange of $A$, $B$ and $C$. For such a case, Eq.~\eqref{EQ:ALINEQ} is clearly not tight as soon as $S(A)=S(B)=S(C)=S(AB)=S(AC)=S(BC)>0$, since the left-hand side of Eq.~\eqref{EQ:ALINEQ} would vanish but its right-hand side would not. On the other hand, in the same case, provided that $\Pi_C^*$ [i.e. the measurement that is optimal for the sake of $\cD_{AB|C}(\rho)$] is such that all conditional states $\rho_{AB}^i$ are separable, Eq.~\eqref{EQ:NEWLEMMA} is tight. This happens, for example, for the tripartite Greenberger-Horne-Zeilinger state $\rho=\proj{\textrm{GHZ}}$, with $\ket{\textrm{GHZ}}=(\ket{000}+\ket{111})/\sqrt{2}$.

\setcounter{theorem}{1}

\begin{theorem}[On entanglement localization]
Let us denote the dimensions of subsystems $A$ and $B$ by $d_A$ and $d_B$, respectively.
If $d_B \ge d_A$ and a tripartite state $\beta = \beta_{ABC}$ has negative partial transposition~\cite{PERES,HoroNPT} in the cut $A:BC$, then it is possible to localize its entanglement onto subsystems $A$ and $B$ using decoding operation on systems $BC$ only.
\end{theorem}

\begin{proof}
We prove that there exists a unitary operation on systems $BC$
followed by a measurement on $C$ and post-selection on a particular outcome,
such that the post-selected state of $AB$ has negative partial transposition.
By assumption there is a pure state $\ket{\psi}$ for which
\begin{equation}
\bra{\psi} \beta^{T_A} \ket{\psi} < 0,
\end{equation}
where $T_A$ denotes partial transposition on system $A$.
The Schmidt decomposition implies $\ket{\psi} = \sum_{j=1}^{d_A} a_j \ket{a_j}_A \ket{\bar a_j}_{BC}$.
By our dimensionality assumption there exists a unitary transformation $U_{BC}$ such that $U_{BC} \ket{\bar a_j}_{BC} \equiv \ket{j}_B \ket{0}_C$.
Therefore, $U_{BC} \ket{\psi} = \ket{\phi}_{AB} \ket{0}_C$ and we have:
\begin{eqnarray}
0 > \bra{\psi} \beta^{T_A} \ket{\psi} & = & \bra{0}\bra{\phi} U_{BC} \beta^{T_A} U_{BC}^\dagger \ket{\phi} \ket{0} \nonumber \\
& = & \bra{\phi} (\bra{0} U_{BC} \beta U_{BC}^\dagger \ket{0} )^{T_A}\ket{\phi},
\end{eqnarray}
where the last equality follows from commutativity of the operations on $A$ and $BC$.
The expression in the bracket is given by the (unnormalized) state $p_0\bar \beta_{AB|0}$ of $AB$ after $C$ observes the measurement result corresponding to the projection on $\ket{0}$. Here $p_0$ is the probability to observe such outcome and $\bar \beta = U_{BC} \beta U_{BC}^\dagger$.
We conclude that $\bar \beta_{AB|0}$ has negative partial transposition.
\end{proof}


\begin{theorem}[On impossibility of entanglement creation via a separable carrier under further restrictions]
\label{LEM_IMPOSSIBLE}
The creation of entanglement through separable states is impossible starting with a qubit $A$ and a qudit $B$ such that $\mathcal{D}_{B|A}(\alpha) = 0$, for $C$ initially in a pure state and unitary encoding.
\end{theorem}
\begin{proof}
By assumption, the initial $AB$ state is of the form
\begin{equation}
{\alpha}_{AB} = \sum_{a=0}^1 p_a \proj{a} \otimes \alpha_{B|a}
\end{equation}
with $\{ \ket{a} \}$ being two generic pure states of A.
The state after unitary $AC$ encoding reads $\beta = \sum_{a=0,1} p_a \proj{\psi_a}_{AC} \otimes \alpha_{B|a}$.
Note that $\alpha_{B|0}$ must be different from $\alpha_{B|1}$ as otherwise distribution is impossible because the initial state is classical on $B$.
Therefore, there exists a POVM element $0\leq M \leq \openone$ such that $\Tr(M \alpha_{B|0})\neq \Tr(M \alpha_{B|1})$. 
This in turns implies that we can choose a family of POVM elements $M(q)$ (e.g., $M(q) = (1-q) M + q \openone$) for which we obtain a family of conditional states of $AC$
\begin{equation}
\label{eq:segment1}
\tilde{\beta}_{AC} = \tilde{p} \proj{\psi_0} + (1-\tilde{p}) \proj{\psi_1},
\end{equation}
with $\tilde{p}$ a probability varying in some finite range. 
Since by assumption, $C$ is separable from $A$ in state $\beta$, 
they are also separable after any measurement on $B$, i.e. $\tilde{\beta}_{AC} $ is separable for all possible $\tilde{p}$.

For any separable $\tilde{\beta}_{AC}$ there exists an ensemble of pure factorized states $\ket{\tilde a_i}\otimes \ket{\tilde c_i}$ with corresponding probabilities $s_i$ such that
\begin{equation}
\label{eq:segment2}
\tilde{\beta}_{AC} =\sum_i s_i \proj{\tilde a_i \tilde c_i}.
\end{equation}
Since $\tilde{\beta}_{AC}$ has rank two, all these pure product states are spanned by two pure product states, let us say $ \ket{\tilde a_0 \tilde c_0}$ and $ \ket{\tilde a_1 \tilde c_1}$, 
which span also $\ket{\psi_0}$ and $\ket{\psi_1}$. It also follows that $\tilde{\beta}_{AC}$ can be seen as a state on $\mathbb{C}^2\otimes\mathbb{C}^2$.

We now use Theorem $1$ of Ref.~\cite{STV} stating that for any plane in $\mathbb{C}^2 \otimes\mathbb{C}^2$ defined by two product vectors, 
either all the states in this plane are product vectors, or there is no other product vector in it.
It follows that either $\ket{\psi_0}$ and $\ket{\psi_1}$ are product vectors
or $\tilde{\beta}_{AC}$ can be written as convex mixture of only $\ket{\tilde a_0 \tilde c_0}$ and $\ket{\tilde a_1 \tilde c_1}$.
Since the space spanned by $\ket{\psi_0}$ and $\ket{\psi_1}$ is the same as that spanned by $\ket{\tilde a_0 \tilde c_0}$ and $\ket{\tilde a_1 \tilde c_1}$,
Eq.~\eqref{eq:segment1} is equal to Eq.~\eqref{eq:segment2} where we sum only over $i=0,1$.
Since this should hold for a finite range of $\tilde{p}$, these two decompositions must coincide and therefore $\ket{\psi_0}$ and $\ket{\psi_1}$ are product vectors.
Finally, $\beta$ is fully separable and entanglement distribution is impossible. 
\end{proof}


\begin{theorem}[On entanglement increase via separable carrier]
Let $d_{\mathrm{tot}}$ denote the total (finite) dimension of the state $\beta_{ABC}$
and let $d$ be the dimension of the carrier system $C$.
If state $\beta_{ABC}$ is $AB:C$ separable, we have
\begin{equation}
\mathcal{E}_{A:CB}(\beta) - \mathcal{E}_{AC:B}(\alpha) \leq \left(1-\frac{1}{d_{\mathrm{tot}}^2} \right) \log d.
\label{SI_ENT_VIA_SEP}
\end{equation}
\end{theorem}

\begin{proof}
Using Eq. (2) of the main text, the left-hand side of \eqref{SI_ENT_VIA_SEP} is upper-bounded by the discord $\mathcal{D}_{AB|C}(\beta)$.
The right-hand side of \eqref{SI_ENT_VIA_SEP} follows from the general relation
\begin{equation}
\cD_{X|Y}(\rho_{X:Y})\leq\left[1-\frac{1}{(d_Xd_Y)^2}\right]\log d_Y,
\label{SI_DISCORD_SEP_BOUND}
\end{equation}
that holds for all separable states $\rho_{X:Y}$ of bipartite systems with (finite) dimensions $d_X$ and $d_Y$, respectively,
by choosing system $Y$ to be the carrier, i.e. $d_Y = d$.
It remains to derive Eq.~\eqref{SI_DISCORD_SEP_BOUND}.
In~Refs.~\cite{DISCORD_ENT,DISCORD_ENT2} it was proven that $\cD_{X|Y}(\tau_{XY})$ can be written as
\begin{equation}
\label{EQ:ACTIV}
\cD_{X|Y}(\tau_{XY})=\min_{\{\ket{y_i}\}}\cE_{XY:Y'}(\tilde{\tau}_{XYY'})
\end{equation}
with $\{\ket{y_i}\}$ an orthonormal basis for $Y$ and $\tilde{\tau}_{XYY'}=\sum_{ij}\tau_X^{ij}\otimes \ket{y_i}\bra{y_j}_{Y}\otimes  \ket{y_i}\bra{y_j}_{Y'}$ if $ \tau_{XY}$ is expanded as $\tau_{XY}= \sum_{ij}\tau_X^{ij}\otimes \ket{y_i}\bra{y_j}_{Y}$. The minimization required in Eq.~\eqref{EQ:ACTIV} is over all possible choices for the basis $\{\ket{y_i}\}$. From such expression for $\cD_{X|Y}(\tau_{XY})$, one finds
$$
\cD_{X|Y}(\tau_{XY})\leq \max_{{\xi}_{XYY'}} \cE_{XY:Y'}(\xi_{XYY'})\leq \log d_Y
$$
due to the fact that the relative entropy of entanglement is upper-bounded by the logarithm of the local dimensions and $d_{Y'}=d_{Y}$. On the other hand, any separable state $\rho_{X:Y}$ admits a pure-state ensemble decomposition $\rho_{X:Y}=\sum_{i=1}^R p_i \proj{\alpha_i}_X \otimes \proj{\beta_i}_Y$ with $R\leq (d_Xd_Y)^2$~\cite{PHRANGE}. Without loss of generality, let $p_1$ be the largest probability in such ensemble. Clearly, $p_1\geq  1/(d_Xd_Y)^2$. By choosing a basis $\{\ket{y_i}\}$ for $Y$ such that $\ket{y_1}=\ket{\beta_1}$, we get
\[
\rho_{X:Y}=p_1 \proj{\alpha_1}_X \otimes \proj{\beta_1}_Y + (1-p_1) \rho'_{X:Y},
\]
with $\rho'_{X:Y}$ a (separable) state. Consequently
\[
\begin{split}
\tilde{\rho}_{XYY'}&=p_1 \proj{\alpha_1}_X \otimes \proj{\beta_1}_Y\otimes \proj{\beta_1}_{Y'} \\
&\quad+ (1-p_1) \tilde{\rho'}_{XYY'}.
\end{split}
\]
We thus find
\begin{multline*}
\cD_{X|Y}(\tau_{XY})\\
\begin{aligned}
&\leq \cE_{XY:Y'}\Big(p_1 \proj{\alpha_1}_X \otimes \proj{\beta_1}_Y\otimes \proj{\beta_1}_{Y'}\\
&\quad+ (1-p_1) \tilde{\rho'}_{XYY'}\Big)\\
&\leq (1-p_1) \cE_{XY:Y'}(\tilde{\rho'}_{XYY'})\\
&\leq \left[1-\frac{1}{(d_Xd_Y)^2}\right]\log d_Y,
\end{aligned}
\end{multline*}
where the first inequality is due to the choice of a specific basis $\{\ket{y_i}\}$, which might be non-optimal for the sake of relation~\eqref{EQ:ACTIV}, and the second inequality comes from the convexity of the relative entropy of entanglement.
\end{proof}

\section{Examples}

We now discuss a series of examples of entanglement distribution via separable carrier states. We consider both entangled and separable initial states.

\begin{example}
Entanglement distribution with vanishing initial entanglement.
\end{example}
We begin with the following Theorem of Ref.~\cite{VT1999}:
\begin{equation}
\rho_{p} \textrm{ is separable} \iff p \le \frac{1}{1+a_1 a_2 d_{\mathrm{tot}}},
\label{VT_THEOREM}
\end{equation}
with $\rho_{p} = p \proj{\psi} + (1-p) \openone/d_{\mathrm{tot}}$;
$a_1,a_2$ are the two biggest Schmidt coefficients of $\ket{\psi}$, 
and $d_{\mathrm{tot}}$ is the dimension of the Hilbert space for the total system.
Consider now a tripartite pure state $\ket{\psi}_{ABC}$.
It admits three Schmidt decompositions corresponding to three different bipartite cuts,
$\ket{\psi} = \sum_i x_i \ket{x_i}_X\ket{\bar{x}_i}_{\widetilde{X}}$ where $x=a,b,c$ denote coefficients and states for suitable bipartitions, 
$X=A,B,C$ denotes a single system, and $\widetilde X = BC,AC,AB$ denotes the other two systems different than $X$.
In order to obtain states useful for entanglement distribution we assume
\begin{equation}
a_1 a_2 > M \equiv \max(b_1 b_2, c_1 c_2).
\label{VT_ASSUMPTION}
\end{equation}
and consider states $\rho_\Upsilon$ for the critical entanglement admixture $\Upsilon = 1/(1+Md_{\mathrm{tot}})$.
By construction, such states $\rho_{\Upsilon}$ satisfy all the requirements of entanglement distribution via separable system.
Furthermore, due to the properties of Schmidt decomposition, the vectors $\ket{\bar b_i}_{AC}$ are orthogonal 
and for $d_B \le d_A$ there exists an encoding unitary such that $U_{AC} \ket{i}_A \ket{0}_C = \ket{\bar b_i}_{AC}$.
Applying this unitary on the initial state $\alpha = \Upsilon \proj{\phi} \otimes \proj{0} + (1-\Upsilon) \frac{1}{d_{\mathrm{tot}}} \openone$,
where the pure state between $A$ and $B$ reads $\ket{\phi} = \sum_i b_i \ket{i}_A \ket{b_i}_B$, produces $\beta = \rho_{\Upsilon}$.
It remains to show that $\alpha$ is $AC:B$-separable, and actually fully separable.
This follows from the fact that the final state $\rho_\Upsilon$ is $AC:B$ separable by construction, and that $\alpha=U^\dagger_{AC}\rho_\Upsilon U_{AC}$. 
Furthermore, $\alpha$ is invariant under a projective measurement on $C$ in the computational basis.
The original example of Ref.~\cite{CVDC} belongs to this family.

\begin{example}
Entanglement distribution with non-zero initial entanglement.
\end{example}

Consider first the three-qubit example studied by Cubitt \emph{et. al.} \cite{CVDC}.
Their initial state reads:
\begin{equation}
\begin{aligned}
\label{cucu}
\Lambda_{ABC} & = \left( \frac{1}{3}|\phi^{+}\rangle \langle\phi^{+}|+\frac{1}{6}|01\rangle\langle01|+\frac{1}{6}|10\rangle \langle10| \right) \otimes |0\rangle \langle0|\\
& + \left( \frac{1}{6}|00\rangle\langle00|+\frac{1}{6}|11\rangle \langle11| \right) \otimes |1\rangle \langle1|,
\end{aligned}
\end{equation}
where $|\phi^{+}\rangle = \frac{1}{\sqrt{2}}(|00\rangle+|11\rangle)$ is the maximally entangled state of $A$ and $B$.
In the original protocol, Alice begins with qubits $A$ and $C$ in her lab, and Bob with qubit $B$.
They do not share any entanglement initially.
Alice then performs a $\textsc{CNOT}$ operation with qubit $A$ as the control qubit and passes qubit $C$ to Bob,
who then performs yet another $\textsc{CNOT}$ operation, this time on the subsystem $BC$ with $B$ as the control qubit.
It was demonstrated that the qubit $C$ is separable at all stages of the protocol
and also that Alice and Bob would share some entanglement in the final state.

Consider the initial state of the form
\begin{equation}
\alpha= p \Lambda_{ABC} + (1-p) \Lambda_{\mathrm{ent}},
\end{equation}
with  $\Lambda_{\mathrm{ent}} \equiv \frac{1}{3}|\phi^{+}\rangle \langle\phi^{+}|\otimes|0\rangle \langle0|+ \left( \frac{1}{3}|00\rangle\langle00|+\frac{1}{3}|11\rangle \langle11| \right)\otimes|1\rangle\langle1|$
and choose $0< p <1$.
This state is then subjected to the same protocol as in the original  example
\begin{equation}
\label{proc}
\alpha \xrightarrow{\textsc{CNOT}_{AC}} \beta \xrightarrow{\textsc{CNOT}_{BC}} \gamma,
\end{equation}
with the final state given by
\begin{eqnarray}
\gamma & = & \frac{1}{3} \proj{\phi^+} \otimes \proj{0} + \frac{2}{3} \gamma_{\mathrm{sep}} \otimes \proj{1},
\end{eqnarray}
where $\gamma_{\mathrm{sep}} = p \openone_4/4 + (1-p)( \proj{00} +\proj{11})/2$
is a separable state of systems $AB$ and $\openone_d$ denotes the $d \times d$ identity matrix.

One can verify using partial transposition that Alice and Bob share some initial entanglement~\cite{PERES,HoroNPT}. 
Furthermore, one has the bound
\begin{equation}
\begin{aligned}
\cE_{AC:B}(\alpha)
&\leq (1-p)\cE_{AC:B}(\Lambda_\mathrm{ent})\\
&=\frac{1-p}{3}\cE_{AC:B}(\proj{\phi^+})\\
&=\frac{1-p}{3}.
\end{aligned}
\end{equation}
The inequality is due to the convexity of the relative entropy of entanglement and the fact that $\Lambda_{ABC}$ is $AC:B$ separable. The first equality is due to the fact that $C$ can be seen as a classical flag in $\Lambda_\mathrm{ent}$~\cite{MHflags}. Finally $\cE_{AC:B}(\proj{\phi^+})=1$.

Putting the state through the entire process summarized in Eq.~\eqref{proc}, one can verify that $\Lambda_{\mathrm{ent}}$ remains $C$-separable after Alice performs her operation
(by checking partial transposition and using the results of Ref.~\cite{DUR}).
The carrier thus remains separable just as in the original example.
At the end of the protocol, if Bob measures the carrier in the standard basis,
he obtains outcome zero with probability ${1}/{3}$, in which case the joint system composed of $A$ and $B$ is in a maximally entangled state.
Therefore, the relative entropy of entanglement of the final state is $\mathcal{E}_{A:BC}(\gamma) ={1}/{3}$~\cite{MHflags}.
Entanglement must have strictly increased through the transfer of a separable carrier.

\begin{example}
Entanglement distribution using initially uncorrelated carrier.
\end{example}
Consider two qubits, $A$ and $B$, prepared in
\begin{equation}
\alpha_{AB} = p \proj{\psi} + (1-p) \tfrac{1}{4} \openone_4,
\end{equation}
where $\ket{\psi} = \sqrt{s} \ket{00} + \sqrt{1-s} \ket{11}$ and $s\in[0,1]$.
The carrier qubit is initially uncorrelated and in the completely mixed state $\alpha_C =\openone_2/2$. 
As unitary encoding operation, we choose
\begin{equation}
U_{AC} = \left(
\begin{array}{cccc}
0 & 0 & 1 & 0 \\
\sqrt{u} & 0 & 0 & -\sqrt{1-u} \\
\sqrt{1-u} & 0 & 0 & \sqrt{u} \\
0 & 1 & 0 & 0
\end{array}
\right),
\end{equation}
with $u \in [0,1]$.
After the application of $U_{AC}$, the initial state is transformed into
$\beta = (\beta_0 +\beta_1)/2$,
where
\begin{equation}
\beta_j = p \proj{\psi_j} + \frac{(1-p)}{8} \openone_8,
\end{equation}
and $\ket{\psi_j} = U_{AC} \ket{\psi} \otimes \ket{j}_C$.
By construction, the entanglement in the $B:AC$ splitting of state $\beta$ vanishes.
To ensure the other separability requirements we apply the Theorem stated in Eq.~\eqref{VT_THEOREM} to the initial state $\alpha_{AB}$ and to the states $\beta_j$. We get
\begin{eqnarray}
\mathcal{E}_{A:B}(\alpha) = 0 & \iff & p \le \frac{1}{1+4 \sqrt{s(1-s)}}, \\
\mathcal{E}_{AB:C}(\beta_0) = 0 & \iff & p \le \frac{1}{1+ 8 \sqrt{su(1-s u)}}, \nonumber \\
\mathcal{E}_{AB:C}(\beta_1) = 0 & \iff & p \le \frac{1}{1+ 8 \sqrt{u(1-s)[1- u (1-s)]}}. \nonumber
\end{eqnarray}
All these requirements are simultaneously satisfied by taking $p = \frac{1}{1+4 \sqrt{s(1-s)}}$ with $u \le 1-\frac{\sqrt{3}}{2} \approx 0.134$ and $s$ in the range
\begin{equation}
\frac{4u(1-u)}{1-4u^2} \le s \le \frac{4u-1}{4u^2 - 1}.
\end{equation}
In order to determine whether $A$ is entangled with $BC$ in state $\beta$, we have resorted to the separability criterion provided by the positivity of the partial transposition~\cite{PERES,HoroNPT}. By taking the partial transpose of the state with respect to $A$, we have determined numerically a range of parameters within which a negative eigenvalue appears. We have found this to occur for $0.022 \gtrsim u > 0$, and the degree of violation of the separability criterion that is maximized for $s = \frac{4u(1-u)}{1-4u^2}$.



\begin{thebibliography}{40}
\expandafter\ifx\csname natexlab\endcsname\relax\def\natexlab#1{#1}\fi
\expandafter\ifx\csname bibnamefont\endcsname\relax
  \def\bibnamefont#1{#1}\fi
\expandafter\ifx\csname bibfnamefont\endcsname\relax
  \def\bibfnamefont#1{#1}\fi
\expandafter\ifx\csname citenamefont\endcsname\relax
  \def\citenamefont#1{#1}\fi
\expandafter\ifx\csname url\endcsname\relax
  \def\url#1{\texttt{#1}}\fi
\expandafter\ifx\csname urlprefix\endcsname\relax\def\urlprefix{URL }\fi
\providecommand{\bibinfo}[2]{#2}
\providecommand{\eprint}[2][]{\url{#2}}

\bibitem[{\citenamefont{Horodecki et~al.}(2009)\citenamefont{Horodecki,
  Horodecki, Horodecki, and Horodecki}}]{ENT_REV}
\bibinfo{author}{\bibfnamefont{R.}~\bibnamefont{Horodecki}},
  \bibinfo{author}{\bibfnamefont{P.}~\bibnamefont{Horodecki}},
  \bibinfo{author}{\bibfnamefont{M.}~\bibnamefont{Horodecki}},
  \bibnamefont{and}
  \bibinfo{author}{\bibfnamefont{K.}~\bibnamefont{Horodecki}},
  \bibinfo{journal}{Rev. Mod. Phys.} \textbf{\bibinfo{volume}{81}},
  \bibinfo{pages}{865} (\bibinfo{year}{2009}).

\bibitem[{\citenamefont{Briegel et~al.}(2009)\citenamefont{Briegel, Browne,
  D\"ur, Raussendorf, and den Nest}}]{ONE-WAY-REVIEW}
\bibinfo{author}{\bibfnamefont{H.~J.} \bibnamefont{Briegel}},
  \bibinfo{author}{\bibfnamefont{D.~E.} \bibnamefont{Browne}},
  \bibinfo{author}{\bibfnamefont{W.}~\bibnamefont{D\"ur}},
  \bibinfo{author}{\bibfnamefont{R.}~\bibnamefont{Raussendorf}},
  \bibnamefont{and} \bibinfo{author}{\bibfnamefont{M.~V.} \bibnamefont{den
  Nest}}, \bibinfo{journal}{Nature Phys.} \textbf{\bibinfo{volume}{5}},
  \bibinfo{pages}{19} (\bibinfo{year}{2009}).

\bibitem[{\citenamefont{Buhrman et~al.}(2010)\citenamefont{Buhrman, Cleve,
  Massar, and de~Wolf}}]{CCP_REV}
\bibinfo{author}{\bibfnamefont{H.}~\bibnamefont{Buhrman}},
  \bibinfo{author}{\bibfnamefont{R.}~\bibnamefont{Cleve}},
  \bibinfo{author}{\bibfnamefont{S.}~\bibnamefont{Massar}}, \bibnamefont{and}
  \bibinfo{author}{\bibfnamefont{R.}~\bibnamefont{de~Wolf}},
  \bibinfo{journal}{Rev. Mod. Phys.} \textbf{\bibinfo{volume}{82}},
  \bibinfo{pages}{665} (\bibinfo{year}{2010}).

\bibitem[{\citenamefont{Ekert}(1991)}]{EKERT1991}
\bibinfo{author}{\bibfnamefont{A.~K.} \bibnamefont{Ekert}},
  \bibinfo{journal}{Phys. Rev. Lett.} \textbf{\bibinfo{volume}{67}},
  \bibinfo{pages}{661} (\bibinfo{year}{1991}).

\bibitem[{\citenamefont{Kimble}(2008)}]{KIMBLE}
\bibinfo{author}{\bibfnamefont{H.~J.} \bibnamefont{Kimble}},
  \bibinfo{journal}{Nature} \textbf{\bibinfo{volume}{453}},
  \bibinfo{pages}{1023} (\bibinfo{year}{2008}).

\bibitem[{\citenamefont{Kraus and Cirac}(2004)}]{TRANSF}
\bibinfo{author}{\bibfnamefont{B.}~\bibnamefont{Kraus}} \bibnamefont{and}
  \bibinfo{author}{\bibfnamefont{J.~I.} \bibnamefont{Cirac}},
  \bibinfo{journal}{Phys. Rev. Lett.} \textbf{\bibinfo{volume}{92}},
  \bibinfo{pages}{013602} (\bibinfo{year}{2004}).

\bibitem[{\citenamefont{Paternostro et~al.}(2004)\citenamefont{Paternostro,
  Son, and Kim}}]{TRANSF2}
\bibinfo{author}{\bibfnamefont{M.}~\bibnamefont{Paternostro}},
  \bibinfo{author}{\bibfnamefont{W.}~\bibnamefont{Son}}, \bibnamefont{and}
  \bibinfo{author}{\bibfnamefont{M.~S.} \bibnamefont{Kim}},
  \bibinfo{journal}{Phys. Rev. Lett.} \textbf{\bibinfo{volume}{92}},
  \bibinfo{pages}{197901} (\bibinfo{year}{2004}).

\bibitem[{\citenamefont{\.Zukowski et~al.}(1993)\citenamefont{\.Zukowski,
  Zeilinger, Horne, and Ekert}}]{ENT_SWAPPING}
\bibinfo{author}{\bibfnamefont{M.}~\bibnamefont{\.Zukowski}},
  \bibinfo{author}{\bibfnamefont{A.}~\bibnamefont{Zeilinger}},
  \bibinfo{author}{\bibfnamefont{M.~A.} \bibnamefont{Horne}}, \bibnamefont{and}
  \bibinfo{author}{\bibfnamefont{A.~K.} \bibnamefont{Ekert}},
  \bibinfo{journal}{Phys. Rev. Lett.} \textbf{\bibinfo{volume}{71}},
  \bibinfo{pages}{4287} (\bibinfo{year}{1993}).

\bibitem[{\citenamefont{Ciccarello et~al.}(2008)\citenamefont{Ciccarello,
  Paternostro, Kim, and Palma}}]{SHUTTLE}
\bibinfo{author}{\bibfnamefont{F.}~\bibnamefont{Ciccarello}},
  \bibinfo{author}{\bibfnamefont{M.}~\bibnamefont{Paternostro}},
  \bibinfo{author}{\bibfnamefont{M.~S.} \bibnamefont{Kim}}, \bibnamefont{and}
  \bibinfo{author}{\bibfnamefont{G.~M.} \bibnamefont{Palma}},
  \bibinfo{journal}{Phys. Rev. Lett.} \textbf{\bibinfo{volume}{100}},
  \bibinfo{pages}{150501} (\bibinfo{year}{2008}).

\bibitem[{\citenamefont{Cubitt et~al.}(2003)\citenamefont{Cubitt, Verstraete,
  D\"ur, and Cirac}}]{CVDC}
\bibinfo{author}{\bibfnamefont{T.~S.} \bibnamefont{Cubitt}},
  \bibinfo{author}{\bibfnamefont{F.}~\bibnamefont{Verstraete}},
  \bibinfo{author}{\bibfnamefont{W.}~\bibnamefont{D\"ur}}, \bibnamefont{and}
  \bibinfo{author}{\bibfnamefont{J.~I.} \bibnamefont{Cirac}},
  \bibinfo{journal}{Phys. Rev. Lett.} \textbf{\bibinfo{volume}{91}},
  \bibinfo{pages}{037902} (\bibinfo{year}{2003}).

\bibitem[{\citenamefont{Mi{\v s}ta et~al.}(2008)\citenamefont{Mi{\v s}ta, Jr.,
  and Korolkova}}]{MK}
\bibinfo{author}{\bibfnamefont{L.}~\bibnamefont{Mi{\v s}ta}},
  \bibinfo{author}{\bibnamefont{Jr.}}, \bibnamefont{and}
  \bibinfo{author}{\bibfnamefont{N.}~\bibnamefont{Korolkova}},
  \bibinfo{journal}{Phys. Rev. A} \textbf{\bibinfo{volume}{77}},
  \bibinfo{pages}{050302(R)} (\bibinfo{year}{2008}).

\bibitem[{\citenamefont{Mi{\v s}ta et~al.}(2009)\citenamefont{Mi{\v s}ta, Jr.,
  and Korolkova}}]{MK2}
\bibinfo{author}{\bibfnamefont{L.}~\bibnamefont{Mi{\v s}ta}},
  \bibinfo{author}{\bibnamefont{Jr.}}, \bibnamefont{and}
  \bibinfo{author}{\bibfnamefont{N.}~\bibnamefont{Korolkova}},
  \bibinfo{journal}{Phys. Rev. A} \textbf{\bibinfo{volume}{80}},
  \bibinfo{pages}{032310} (\bibinfo{year}{2009}).

\bibitem[{\citenamefont{Bennett et~al.}(1996)\citenamefont{Bennett, DiVincenzo,
  Smolin, and Wootters}}]{LOCC}
\bibinfo{author}{\bibfnamefont{C.~H.} \bibnamefont{Bennett}},
  \bibinfo{author}{\bibfnamefont{D.~P.} \bibnamefont{DiVincenzo}},
  \bibinfo{author}{\bibfnamefont{J.~A.} \bibnamefont{Smolin}},
  \bibnamefont{and} \bibinfo{author}{\bibfnamefont{W.~K.}
  \bibnamefont{Wootters}}, \bibinfo{journal}{Phys. Rev. A}
  \textbf{\bibinfo{volume}{54}}, \bibinfo{pages}{3824} (\bibinfo{year}{1996}).

\bibitem[{\citenamefont{Datta}(2008)}]{DATTA_THESIS}
\bibinfo{author}{\bibfnamefont{A.}~\bibnamefont{Datta}}, Ph.D. thesis,
  \bibinfo{school}{University of New Mexico} (\bibinfo{year}{2008}),
  \eprint{arXiv:0807.4490}.

\bibitem[{\citenamefont{Li and Luo}(2008)}]{LL2008}
\bibinfo{author}{\bibfnamefont{N.}~\bibnamefont{Li}} \bibnamefont{and}
  \bibinfo{author}{\bibfnamefont{S.}~\bibnamefont{Luo}},
  \bibinfo{journal}{Phys. Rev. A} \textbf{\bibinfo{volume}{78}},
  \bibinfo{pages}{024303} (\bibinfo{year}{2008}).

\bibitem[{\citenamefont{Henderson and Vedral}(2001)}]{HV}
\bibinfo{author}{\bibfnamefont{L.}~\bibnamefont{Henderson}} \bibnamefont{and}
  \bibinfo{author}{\bibfnamefont{V.}~\bibnamefont{Vedral}},
  \bibinfo{journal}{J. Phys. A} \textbf{\bibinfo{volume}{34}},
  \bibinfo{pages}{6899} (\bibinfo{year}{2001}).

\bibitem[{\citenamefont{Ollivier and Zurek}(2002)}]{OZ}
\bibinfo{author}{\bibfnamefont{H.}~\bibnamefont{Ollivier}} \bibnamefont{and}
  \bibinfo{author}{\bibfnamefont{W.~H.} \bibnamefont{Zurek}},
  \bibinfo{journal}{Phys. Rev. Lett.} \textbf{\bibinfo{volume}{88}},
  \bibinfo{pages}{017901} (\bibinfo{year}{2002}).

\bibitem[{\citenamefont{Modi et~al.}(2011)\citenamefont{Modi, Brodutch, Cable,
  Paterek, and Vedral}}]{DISCORD_REV}
\bibinfo{author}{\bibfnamefont{K.}~\bibnamefont{Modi}},
  \bibinfo{author}{\bibfnamefont{A.}~\bibnamefont{Brodutch}},
  \bibinfo{author}{\bibfnamefont{H.}~\bibnamefont{Cable}},
  \bibinfo{author}{\bibfnamefont{T.}~\bibnamefont{Paterek}}, \bibnamefont{and}
  \bibinfo{author}{\bibfnamefont{V.}~\bibnamefont{Vedral}}
  (\bibinfo{year}{2011}), \eprint{arXiv:1112.6238}.

\bibitem[{\citenamefont{Horodecki
  et~al.}(2005{\natexlab{a}})\citenamefont{Horodecki, Oppenheim, and
  Winter}}]{MERGING}
\bibinfo{author}{\bibfnamefont{M.}~\bibnamefont{Horodecki}},
  \bibinfo{author}{\bibfnamefont{J.}~\bibnamefont{Oppenheim}},
  \bibnamefont{and} \bibinfo{author}{\bibfnamefont{A.}~\bibnamefont{Winter}},
  \bibinfo{journal}{Nature} \textbf{\bibinfo{volume}{436}},
  \bibinfo{pages}{673} (\bibinfo{year}{2005}{\natexlab{a}}).

\bibitem[{\citenamefont{del Rio et~al.}(2011)\citenamefont{del Rio, Aberg,
  Renner, Dahlsten, and Vedral}}]{THERMO_S}
\bibinfo{author}{\bibfnamefont{L.}~\bibnamefont{del Rio}},
  \bibinfo{author}{\bibfnamefont{J.}~\bibnamefont{Aberg}},
  \bibinfo{author}{\bibfnamefont{R.}~\bibnamefont{Renner}},
  \bibinfo{author}{\bibfnamefont{O.}~\bibnamefont{Dahlsten}}, \bibnamefont{and}
  \bibinfo{author}{\bibfnamefont{V.}~\bibnamefont{Vedral}},
  \bibinfo{journal}{Nature} \textbf{\bibinfo{volume}{474}}, \bibinfo{pages}{61}
  (\bibinfo{year}{2011}).

\bibitem[{\citenamefont{Vedral et~al.}(1997)\citenamefont{Vedral, Plenio,
  Rippin, and Knight}}]{REL_ENT_E}
\bibinfo{author}{\bibfnamefont{V.}~\bibnamefont{Vedral}},
  \bibinfo{author}{\bibfnamefont{M.~B.} \bibnamefont{Plenio}},
  \bibinfo{author}{\bibfnamefont{M.~A.} \bibnamefont{Rippin}},
  \bibnamefont{and} \bibinfo{author}{\bibfnamefont{P.~L.}
  \bibnamefont{Knight}}, \bibinfo{journal}{Phys. Rev. Lett.}
  \textbf{\bibinfo{volume}{78}}, \bibinfo{pages}{2275} (\bibinfo{year}{1997}).

\bibitem[{\citenamefont{Vedral and Plenio}(1998)}]{REL_ENT_E2}
\bibinfo{author}{\bibfnamefont{V.}~\bibnamefont{Vedral}} \bibnamefont{and}
  \bibinfo{author}{\bibfnamefont{M.~B.} \bibnamefont{Plenio}},
  \bibinfo{journal}{Phys. Rev. A} \textbf{\bibinfo{volume}{57}},
  \bibinfo{pages}{1619} (\bibinfo{year}{1998}).

\bibitem[{\citenamefont{Horodecki
  et~al.}(2005{\natexlab{b}})\citenamefont{Horodecki, Horodecki, Horodecki,
  Oppenheim, Sen(De), Sen, and Synak-Radtke}}]{DEFICIT}
\bibinfo{author}{\bibfnamefont{M.}~\bibnamefont{Horodecki}},
  \bibinfo{author}{\bibfnamefont{P.}~\bibnamefont{Horodecki}},
  \bibinfo{author}{\bibfnamefont{R.}~\bibnamefont{Horodecki}},
  \bibinfo{author}{\bibfnamefont{J.}~\bibnamefont{Oppenheim}},
  \bibinfo{author}{\bibfnamefont{A.}~\bibnamefont{Sen(De)}},
  \bibinfo{author}{\bibfnamefont{U.}~\bibnamefont{Sen}}, \bibnamefont{and}
  \bibinfo{author}{\bibfnamefont{B.}~\bibnamefont{Synak-Radtke}},
  \bibinfo{journal}{Phys. Rev. A} \textbf{\bibinfo{volume}{71}},
  \bibinfo{pages}{062307} (\bibinfo{year}{2005}{\natexlab{b}}).

\bibitem[{\citenamefont{Modi et~al.}(2010)\citenamefont{Modi, Paterek, Son,
  Vedral, and Williamson}}]{MODI2010}
\bibinfo{author}{\bibfnamefont{K.}~\bibnamefont{Modi}},
  \bibinfo{author}{\bibfnamefont{T.}~\bibnamefont{Paterek}},
  \bibinfo{author}{\bibfnamefont{W.}~\bibnamefont{Son}},
  \bibinfo{author}{\bibfnamefont{V.}~\bibnamefont{Vedral}}, \bibnamefont{and}
  \bibinfo{author}{\bibfnamefont{M.}~\bibnamefont{Williamson}},
  \bibinfo{journal}{Phys. Rev. Lett.} \textbf{\bibinfo{volume}{104}},
  \bibinfo{pages}{080501} (\bibinfo{year}{2010}).

\bibitem[{\citenamefont{Groisman et~al.}(2005)\citenamefont{Groisman, Popescu,
  and Winter}}]{GPW2005}
\bibinfo{author}{\bibfnamefont{B.}~\bibnamefont{Groisman}},
  \bibinfo{author}{\bibfnamefont{S.}~\bibnamefont{Popescu}}, \bibnamefont{and}
  \bibinfo{author}{\bibfnamefont{A.}~\bibnamefont{Winter}},
  \bibinfo{journal}{Phys. Rev. A} \textbf{\bibinfo{volume}{72}},
  \bibinfo{pages}{032317} (\bibinfo{year}{2005}).

\bibitem[{\citenamefont{Verstraete et~al.}(2004)\citenamefont{Verstraete, Popp,
  and Cirac}}]{ENT_LOC}
\bibinfo{author}{\bibfnamefont{F.}~\bibnamefont{Verstraete}},
  \bibinfo{author}{\bibfnamefont{M.}~\bibnamefont{Popp}}, \bibnamefont{and}
  \bibinfo{author}{\bibfnamefont{J.~I.} \bibnamefont{Cirac}},
  \bibinfo{journal}{Phys. Rev. Lett.} \textbf{\bibinfo{volume}{92}},
  \bibinfo{pages}{027901} (\bibinfo{year}{2004}).

\bibitem[{\citenamefont{Araki and Lieb}(1970)}]{ARAKI-LIEB}
\bibinfo{author}{\bibfnamefont{H.}~\bibnamefont{Araki}} \bibnamefont{and}
  \bibinfo{author}{\bibfnamefont{E.~H.} \bibnamefont{Lieb}},
  \bibinfo{journal}{Comm. Math. Phys.} \textbf{\bibinfo{volume}{18}},
  \bibinfo{pages}{160} (\bibinfo{year}{1970}).

\bibitem[{\citenamefont{Horodecki et~al.}(1998)\citenamefont{Horodecki,
  Horodecki, and Horodecki}}]{BOUND_ENT}
\bibinfo{author}{\bibfnamefont{M.}~\bibnamefont{Horodecki}},
  \bibinfo{author}{\bibfnamefont{P.}~\bibnamefont{Horodecki}},
  \bibnamefont{and}
  \bibinfo{author}{\bibfnamefont{R.}~\bibnamefont{Horodecki}},
  \bibinfo{journal}{Phys. Rev. Lett.} \textbf{\bibinfo{volume}{80}},
  \bibinfo{pages}{5239} (\bibinfo{year}{1998}).

\bibitem[{\citenamefont{Horodecki et~al.}(1997)\citenamefont{Horodecki,
  Horodecki, and Horodecki}}]{2QUBIT_DIST}
\bibinfo{author}{\bibfnamefont{M.}~\bibnamefont{Horodecki}},
  \bibinfo{author}{\bibfnamefont{P.}~\bibnamefont{Horodecki}},
  \bibnamefont{and}
  \bibinfo{author}{\bibfnamefont{R.}~\bibnamefont{Horodecki}},
  \bibinfo{journal}{Phys. Rev. Lett.} \textbf{\bibinfo{volume}{78}},
  \bibinfo{pages}{574} (\bibinfo{year}{1997}).

\bibitem[{\citenamefont{Streltsov et~al.}(2011)\citenamefont{Streltsov,
  Kampermann, and Bru\ss}}]{DISCORD_ENT}
\bibinfo{author}{\bibfnamefont{A.}~\bibnamefont{Streltsov}},
  \bibinfo{author}{\bibfnamefont{H.}~\bibnamefont{Kampermann}},
  \bibnamefont{and} \bibinfo{author}{\bibfnamefont{D.}~\bibnamefont{Bru\ss}},
  \bibinfo{journal}{Phys. Rev. Lett.} \textbf{\bibinfo{volume}{106}},
  \bibinfo{pages}{160401} (\bibinfo{year}{2011}).

\bibitem[{\citenamefont{Gurvits and Barnum}(2002)}]{GB}
\bibinfo{author}{\bibfnamefont{L.}~\bibnamefont{Gurvits}} \bibnamefont{and}
  \bibinfo{author}{\bibfnamefont{H.}~\bibnamefont{Barnum}},
  \bibinfo{journal}{Phys. Rev. A} \textbf{\bibinfo{volume}{66}},
  \bibinfo{pages}{062311} (\bibinfo{year}{2002}).

\bibitem[{\citenamefont{Ferraro et~al.}(2010)\citenamefont{Ferraro, Aolita,
  Cavalcanti, Cucchietti, and Acin}}]{FACCA}
\bibinfo{author}{\bibfnamefont{A.}~\bibnamefont{Ferraro}},
  \bibinfo{author}{\bibfnamefont{L.}~\bibnamefont{Aolita}},
  \bibinfo{author}{\bibfnamefont{D.}~\bibnamefont{Cavalcanti}},
  \bibinfo{author}{\bibfnamefont{F.~M.} \bibnamefont{Cucchietti}},
  \bibnamefont{and} \bibinfo{author}{\bibfnamefont{A.}~\bibnamefont{Acin}},
  \bibinfo{journal}{Phys. Rev. A} \textbf{\bibinfo{volume}{81}},
  \bibinfo{pages}{052318} (\bibinfo{year}{2010}).

\bibitem[{\citenamefont{Vidal and Tarrach}(1999)}]{VT1999}
\bibinfo{author}{\bibfnamefont{G.}~\bibnamefont{Vidal}} \bibnamefont{and}
  \bibinfo{author}{\bibfnamefont{R.}~\bibnamefont{Tarrach}},
  \bibinfo{journal}{Phys. Rev. A} \textbf{\bibinfo{volume}{59}},
  \bibinfo{pages}{141} (\bibinfo{year}{1999}).

\bibitem[{\citenamefont{Paw{\l}owski et~al.}(2009)\citenamefont{Paw{\l}owski,
  Paterek, Kaszlikowski, Scarani, Winter, and \.Zukowski}}]{IC}
\bibinfo{author}{\bibfnamefont{M.}~\bibnamefont{Paw{\l}owski}},
  \bibinfo{author}{\bibfnamefont{T.}~\bibnamefont{Paterek}},
  \bibinfo{author}{\bibfnamefont{D.}~\bibnamefont{Kaszlikowski}},
  \bibinfo{author}{\bibfnamefont{V.}~\bibnamefont{Scarani}},
  \bibinfo{author}{\bibfnamefont{A.}~\bibnamefont{Winter}}, \bibnamefont{and}
  \bibinfo{author}{\bibfnamefont{M.}~\bibnamefont{\.Zukowski}},
  \bibinfo{journal}{Nature} \textbf{\bibinfo{volume}{461}},
  \bibinfo{pages}{1101} (\bibinfo{year}{2009}).

\bibitem[{\citenamefont{Nielsen and Chuang}(2000)}]{NC}
\bibinfo{author}{\bibfnamefont{M.~A.} \bibnamefont{Nielsen}} \bibnamefont{and}
  \bibinfo{author}{\bibfnamefont{I.~L.} \bibnamefont{Chuang}},
  \emph{\bibinfo{title}{Quantum Computation and Quantum Information}}
  (\bibinfo{publisher}{Cambridge}, \bibinfo{year}{2000}).


\bibitem[{\citenamefont{Streltsov et~al.}(2012)\citenamefont{Streltsov,
  Kampermann, and Bru\ss}}]{ALEXANDUS}
\bibinfo{author}{\bibfnamefont{A.}~\bibnamefont{Streltsov}},
  \bibinfo{author}{\bibfnamefont{H.}~\bibnamefont{Kampermann}},
  \bibnamefont{and} \bibinfo{author}{\bibfnamefont{D.}~\bibnamefont{Bru\ss}},
\bibinfo{journal}{Phys. Rev. Lett.} \textbf{\bibinfo{volume}{108}},
  \bibinfo{pages}{250501} (\bibinfo{year}{2012}).

\bibitem[{\citenamefont{Piani}(2009)}]{restrictedmeasurements}
\bibinfo{author}{\bibfnamefont{M.}~\bibnamefont{Piani}},
  \bibinfo{journal}{Phys. Rev. Lett.} \textbf{\bibinfo{volume}{103}},
  \bibinfo{pages}{160504} (\bibinfo{year}{2009}).

\bibitem[{\citenamefont{Horodecki}(2005)}]{MHflags}
\bibinfo{author}{\bibfnamefont{M.}~\bibnamefont{Horodecki}},
  \bibinfo{journal}{Open Syst. Inf. Dyn.} \textbf{\bibinfo{volume}{12}},
  \bibinfo{pages}{231} (\bibinfo{year}{2005}).

\bibitem[{\citenamefont{Horodecki
  et~al.}(2005{\natexlab{c}})\citenamefont{Horodecki, Horodecki, Horodecki, and
  Oppenheim}}]{LOCKINGENT}
\bibinfo{author}{\bibfnamefont{K.}~\bibnamefont{Horodecki}},
  \bibinfo{author}{\bibfnamefont{M.}~\bibnamefont{Horodecki}},
  \bibinfo{author}{\bibfnamefont{P.}~\bibnamefont{Horodecki}},
  \bibnamefont{and}
  \bibinfo{author}{\bibfnamefont{J.}~\bibnamefont{Oppenheim}},
  \bibinfo{journal}{Phys. Rev. Lett.} \textbf{\bibinfo{volume}{94}},
  \bibinfo{pages}{200501} (\bibinfo{year}{2005}{\natexlab{c}}).

\bibitem{PERES} A. Peres, Phys. Rev. Lett. {\bf 77}, 1413 (1996).

\bibitem{HoroNPT} M. Horodecki, P. Horodecki, and R. Horodecki, Phys. Lett. A 
{\bf 223}, 1 (1996).

\bibitem{STV}  A. Sanpera, R. Tarrach, and G. Vidal, Phys. Rev. A {\bf 58}, 826 (1998). 

\bibitem{DISCORD_ENT2} M. Piani, S. Gharibian, G. Adesso, J. Calsamiglia, 
P. Horodecki, and A. Winter, Phys. Rev. Lett. {\bf 106}, 220403 
(2011).

\bibitem{PHRANGE} P. Horodecki, Phys. Lett. A {\bf 232}, 333 (1997).

\bibitem{DUR} W. D\" ur, J. I. Cirac, M. Lewenstein, and D. Bru\ss, Phys. Rev. A 
{\bf 61}, 062313 (2000).

\end{thebibliography}
\end{document}